\documentclass[conference, a4paper]{IEEEtran}
\usepackage[left=13mm,right=13mm,top=19mm,bottom=43mm]{geometry}
\usepackage{epsf}
\usepackage{graphicx}
\usepackage{amsmath,bm}
\usepackage{amssymb}
\usepackage{epsfig,latexsym,amsmath,epsf,amssymb,amsfonts}
\usepackage{verbatim}
\usepackage{placeins}
\usepackage[hang]{subfigure}
\usepackage{epstopdf}
\usepackage{multirow}
\usepackage{cite} 
\usepackage{amsthm}
\usepackage{url}
\usepackage{cases}
\usepackage[utf8]{inputenc}
\usepackage{chngcntr}
\usepackage{tabu}
\usepackage[linesnumbered,ruled]{algorithm2e}
\usepackage[justification=centering]{caption}
\usepackage{siunitx}

\allowdisplaybreaks
\usepackage{color}
\newtheorem{lma}{Lemma}
\newtheorem{thm}{Theorem}
\newtheorem{corl}{Corollary}
\newtheorem{remarkl}{Remark}
\DeclareMathOperator{\E}{\mathbb{E}}
\DeclareMathOperator{\Prb}{\mathbb{P}}

\begin{document}

\title{Coverage and Rate Analysis for Unmanned Aerial Vehicle Base Stations with LoS/NLoS Propagation}
\author{Mohamed Alzenad and~Halim~Yanikomeroglu\\
Department of Systems and Computer Engineering, Carleton University, Ottawa, ON, Canada\\
Email: \{mohamed.alzenad, halim\}@sce.carleton.ca
}

\maketitle

\begin{abstract}
The use of unmanned aerial vehicle base stations (UAV-BSs) as airborne base stations has recently gained great attention. In this paper, we model a network of UAV-BSs as a Poisson point process (PPP) operating at a certain altitude above the ground users. We adopt an air-to-ground (A2G) channel model that incorporates line-of-sight (LoS) and non-line-of-sight (NLoS) propagation. Thus, UAV-BSs can be decomposed into two independent inhomogeneous PPPs. 
 Under the assumption that NLoS and LoS channels experience Rayleigh and Nakagami-m fading, respectively, we derive approximations for the coverage probability and average achievable rate, and show that these approximations match the simulations with negligible errors. Numerical simulations have shown that the coverage probability and average achievable rate decrease as the height of the UAV-BSs increases.
\end{abstract}


\section{introduction} 
Flexible and easy-to-deploy solutions to provide wireless connectivity are of vital importance in current and future wireless systems. Therefore, 
the use of unmanned aerial vehicle base stations (UAV-BSs) to enhance coverage or boost capacity has recently attracted great attention \cite{cao2018airborne,irem2018spatial}. UAV-BSs can assist the terrestrial wireless network in a variety of scenarios. 
For example, UAV-BSs can be quickly deployed during the aftermath of a natural disaster or to offload traffic from a congested terrestrial BS during a sports event\cite{alzenad2018fso,zeng2016wireless,IremMagazine}. 
Recently, there have been several works on UAV-BS deployment, e.g.,\cite{alzenad20173,alzenad20173d,Elham2017Backhaul,MozaffariEfficient}. The authors in \cite{alzenad20173} proposed a framework for evaluating the 3D location of the UAV-BS that maximizes the number of covered users using minimum transmit power while the work in \cite{alzenad20173d} investigated the 3D placement problem for different QoS requirements. In \cite{Elham2017Backhaul}, the authors developed a grid search algorithm to address a backhaul-aware 3D UAV-BS placement problem.  A framework for 3D UAV-BSs deployment based on circle packing was proposed in \cite{MozaffariEfficient}. Moreover, the authors in \cite{MozaffariEfficient} derived the coverage probability as a function of altitude and antenna gain. However, the work in \cite{alzenad20173,alzenad20173d,Elham2017Backhaul,MozaffariEfficient} aimed at finding the exact 3D location which may be unnecessary and difficult to obtain. 

Stochastic geometry has been widely used to model and analyze terrestrial wireless networks. However, a handful of works adopted such approach for UAV-assisted networks. An exact analytical expression for the coverage probability of uniformly distributed UAV-BSs was derived in \cite{galkin2017coverage}. This work adopted a terrestrial channel model for the A2G channels and assumed that all wireless links are subject to Nakagami-m fading. The work in \cite{chetlur2017downlink} modeled the UAV-BSs as a 2D Binomial point process (BPP) in a disc located at a fixed altitude. The authors assumed that all the UAV-BSs are in LoS condition with the users and hence Nakagami-m fading was assumed for all wireless links. Additionally, the thermal noise was assumed negligible in comparison to interference (interference-limited scenario). The exact coverage probability, and accurate coverage probability approximation for Nakagam-m and fading-free channels were also derived.  
The authors in \cite{zhang2017spectrum} investigated spectrum sharing between UAV-BSs and a terrestrial cellular network using tools from stochastic geometry. The UAV-BSs were modeled as a 3D PPP with a minimum height while the terrestrial cellular network was assumed to form a 2D PPP. Additionally in \cite{zhang2017spectrum}, it was assumed that all the UAV-BSs undergo Rayleigh fading which is justified for NLoS transmissions. A network comprised of a single terrestrial BS and a single UAV-BS was investigated in \cite{zhou2018uplink}. The authors derived analytical expressions for the uplink coverage probability of a terrestrial BS and a UAV-BS.

\textit{Contributions}: We adopt an A2G channel model that captures both LoS and NLoS transmissions. Although Rayleigh fading assumption is common for NLoS channels, it may not be for LoS channels. 
Therefore, we adopt the Nakagami-m distribution for LoS channels. We derive the distribution of the distances from the typical user to the closest NLoS and LoS UAV-BSs. 
 After that, we derive a closed-form expression for the Laplace transform of the aggregated interference power as a  function of the altitude and density of UAV-BSs. Unlike the works \cite{galkin2017coverage,chetlur2017downlink} in which the evaluation of the coverage probability involves finding m numerical derivatives of the Laplace transform of the interference, we derive tractable approximations for the coverage probability and average achievable data rate using bounds on incomplete Gamma function. We show that the approximate coverage probability and average achievable data rate match the simulations very closely. 

\section{system model}
We consider a network of UAV-BSs and focus on the analysis of the downlink performance. The UAV-BSs are assumed to be uniformly distributed on an infinite plane located at some altitude $h$ [m] as depicted in Fig. \ref{fig:SystemModel}. We assume that the UAV-BSs form a homogeneous PPP, denoted by $\Phi\overset{\Delta}{=}\{x_i\}$, with density $\lambda$ [BS/$\textup {km}^2$] where $x_i$ refers to the 3D location of the UAV-BS $i$. Also, we assume that all the UAV-BSs transmit at the same power $P_t$ and 
a frequency reuse of 1 is used. This implies that the UAV-BSs interfere with each other. However, within a cell, we assume that orthogonal transmission is implemented which implies that intra-cell interference does not occur. Thus, the typical user does not receive interference signals from its serving BS. Without loss of generality, we consider a typical user located at the origin $O$.
\begin{table}[t!]
\renewcommand{\arraystretch}{1.2}
\caption{{Notation and Symbols Summary}}
\label{Table:Notation}
\newcommand{\tabincell}[2]{\begin{tabular}{@{}#1@{}}#2\end{tabular}}
\centering
\resizebox{\columnwidth}{!}{
\begin{tabular}{|c|c|}
\hline
{Notation} & {Description}  \\
\hline
{PPP} &{\tabincell{c}{Poisson point process}}  \\
\hline
{A2G} &{\tabincell{c}{Air-to-Ground}}  \\
\hline
{h} & {\tabincell{c}{Height of UAV-BSs}}  \\
\cline{1-2}
{$\Phi;\lambda$} &{\tabincell{c}{PPP of UAV-BSs; density of UAV-BSs}}  \\
\hline
{$x_i;x_o$} &{\tabincell{c}{3D location of UAV-BS $i$; 3D location of serving UAV-BS}}  \\
\hline
{$\Phi^N;\Phi^L$} &{\tabincell{c}{PPP of NLoS UAV-BSs; PPP of LoS UAV-BSs}}  \\
\hline
{$P_N(z);P_L(z)$} & {Probability of NLoS; probability of LoS} \\
\hline
{$m$} & {\tabincell{c}{ Parameter of Nakagami-m distribution for LoS links}}  \\
\hline
{$D_{N,x_i},D_{L,x_i}$} &  {\tabincell{c}{Distance between the typical user and a NLoS UAV-BS,\\ or a LoS UAV-BS located at point $x_i$, respectively}} \\
\hline
{$H_{x_i}$} &  {\tabincell{c}{Channel power gain between the typical user \\and a NLoS  UAV-BS located at point $x_i$}} \\
\hline
{$G_{x_i}$} &  {\tabincell{c}{Channel power gain between the typical user \\and a LoS UAV-BS located at point $x_i$}} \\
\hline
{$\alpha_N,\alpha_L$} & {\tabincell{c}{ Path loss exponent for NLoS, and LoS links, respectively }} \\
\hline
{$P_t;\sigma^2$} &{\tabincell{c}{Transmit power of UAV-BSs; thermal noise power}}  \\
\hline
{$\eta_N,\eta_L$} &{\tabincell{c}{Additional losses for NLoS, and LoS links, respectively}}  \\
\hline
{$R_N,R_L$} & {\tabincell{c}{Distance between the typical user and the closest \\NLoS, and LoS UAV-BS, respectively}}  \\
\hline
{$f_{R_N}(r),f_{R_L}(r)$} & {\tabincell{c}{Distribution of the distance between the typical user and \\ the closest NLoS, and LoS UAV-BS, respectively}}  \\
\hline
{$I;\mathcal{L}_I(s|r)$} & {\tabincell{c}{ Interference; Laplace transform of interference at $s$}} \\
\hline
{$A_N,A_L$} &{\tabincell{c}{Probability that the typical user is associated with \\a NLoS UAV-BS, or a LoS UAV-BS, respectively}}  \\
\hline
{$P_C;\tau$} & {Probability of coverage; average downlink rate}  \\
\hline
{$P_{C,N},P_{C,L}$} &{\tabincell{c}{Coverage probability given that the typical user is associated \\with a NLoS, or a LoS UAV-BS, respectively}}  \\
\hline
{$T$} & {\tabincell{c}{\textsf{SINR} threshold for successful communication}}  \\
\hline
{$\tau_N,\tau_L$} &{\tabincell{c}{Average rate given that the typical user is associated with\\ a NLoS, or a LoS UAV-BS, respectively}}  \\
\hline
\end{tabular}
}
\end{table}  
\subsection{Channel Model}
The links between the UAV-BSs and the ground users are mainly LoS or NLoS \cite{Hourani}. For a given altitude $h$, the occurrence of LoS and NLoS transmissions can be captured using the probability of LoS transmission, denoted by $P_L(z)$, and the probability of NLoS transmission, denoted by  $P_N(z)$, where \cite{Hourani} 
 
 \begin{equation}\label{Eq:Prob}
P_L(z)=\frac{1}{1+a \exp (-b(\frac{180}{\pi}\tan^{-1}(\frac{h}{z}) -a))},
\end{equation} 
where $a$ and $b$ are constants that depend on the environment, and $z$ denotes the Euclidean horizontal distance between the typical user and the projection of the UAV-BS location on the horizontal plane. Furthermore, the probability of NLoS is $P_N(z)=1-P_L(z)$.
\begin{figure}
\begin{center}
\includegraphics[ height=5cm, width=8cm]{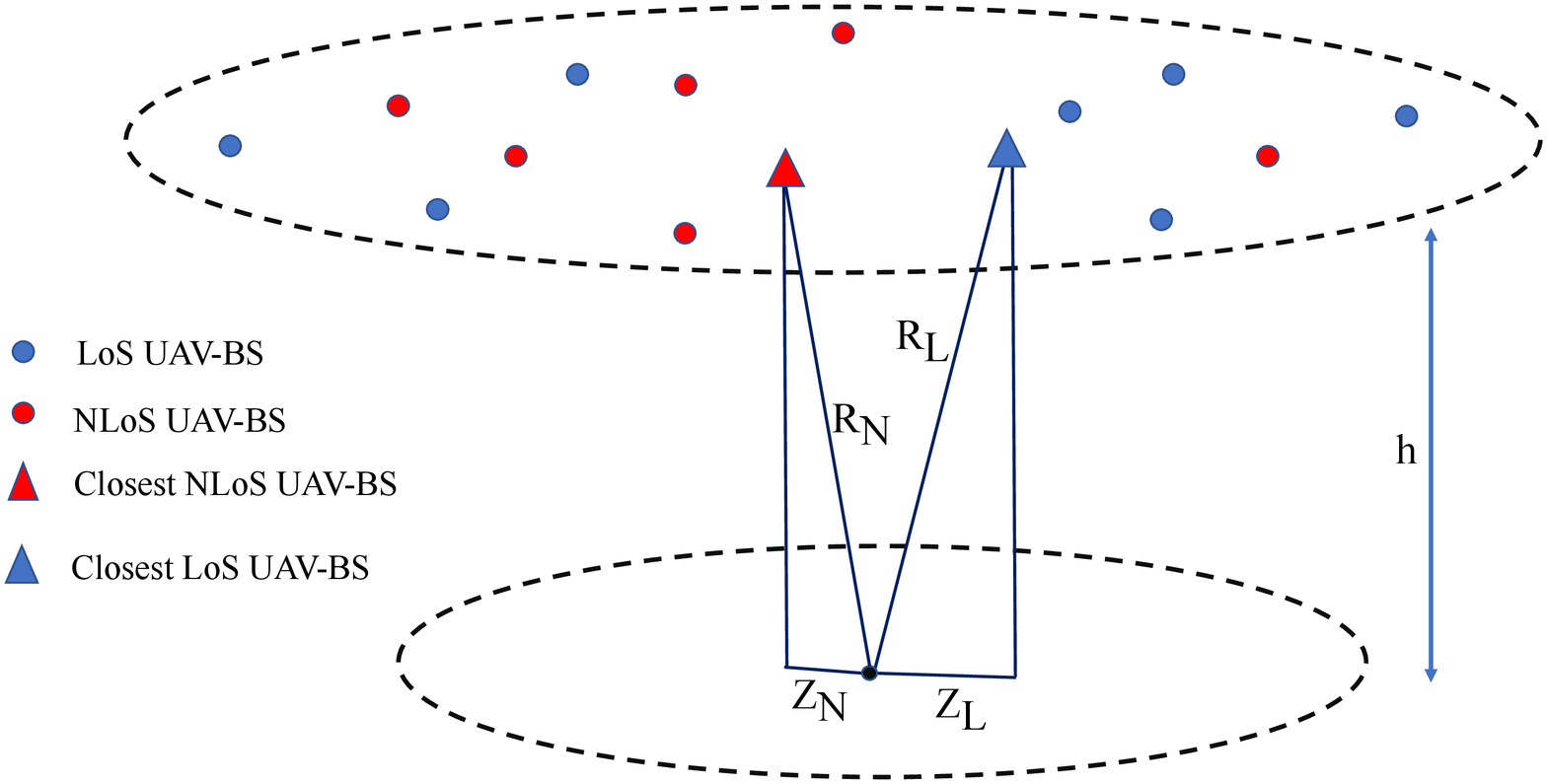}
\caption{\footnotesize Illustration of the system model.}
\label{fig:SystemModel}
\end{center}
\end{figure}

In our model, we assume that each UAV-BS is either in a LoS or NLoS condition with the typical user and that LoS and NLoS transmissions are independent from each other. This implies that the set of UAV-BSs $\Phi$ can be decomposed into two independent inhomogeneous PPPs, i.e., $\Phi=\Phi^L\cup \Phi^N$, where $\Phi^L$ and $\Phi^N$ denote the set of LoS and NLoS UAV-BSs, respectively. Note that the resultant PPPs ($\Phi^L$ and $\Phi^N$) are inhomogeneous because $P_L(z)$ and $P_N(z)$ are functions of $z$. Clearly, for a given altitude $h$, a UAV-BS with a large horizontal distance $z$ is more likely to be in a NLoS condition with the typical user. 

We assume that NLoS and LoS transmissions are characterized by different small scale fading. In particular, we assume that the fading loss, denoted by $H_{x_i}$, between a NLoS UAV-BS located at point $x_i$ and the typical user is exponentially distributed (Rayleigh fading), i.e., $H_{x_i} \sim \exp(1), \forall x_i\in \Phi^N$. 
For LoS transmissions, we choose the well known Nakagami-m distribution  with the shape parameter $m$ which can capture a wide range of fading scenarios. As a result, the channel fading power gain for LoS links, denoted by $G_{x_i}$, follows Gamma distribution with probability density function given by \cite{wackerly2007mathematical}

 \begin{equation}\label{Eq:GammaDistr}
f_{G_{x_i}}(g)=\frac{m^mg^{m-1}}{\Gamma(m)}e^{-mg}, \quad \forall x_i\in\Phi^L,
\end{equation} 
where $\Gamma(m)$ is the Gamma function given by ${\Gamma(m)=\int_{0}^{\infty} x^{m-1}e^{-x}dx}$.

Let $\eta_{\textup N}$ and $\eta_{\textup L}$ denote the mean additional losses for NLoS and LoS transmissions, respectively \cite{Hourani}. The received power at the typical user from a UAV-BS located at point $x_i$ is given by
\begin{equation}\label{Eq:LoSNLoSRecPow}
P_{x_i}=
\begin{cases}
  \zeta_N H_{x_i} D_{N,x_i}^{-\alpha_{N}},  \hspace{10pt} \forall x_i\in\Phi^N\nonumber \\
  \zeta_L G_{x_i} D_{L,x_i}^{-\alpha_{L}}\hspace{3pt}, \hspace{10pt} \forall x_i\in\Phi^L,
\end{cases}
\end{equation}
where $\zeta_N=P_t\eta_N$, and $\zeta_L=P_t\eta_L$. Also, $D_{N,x_i}$ and $D_{L,x_i}$ are the distances between a UAV-BS located at point $x_i$ and the typical user for NLoS and LoS transmissions, respectively. Finally, $\alpha_N$ and $\alpha_L$ are the path loss exponents for NLoS and LoS transmissions, respectively.
The notation and symbols used in this paper are summarized in Table \ref{Table:Notation}.
\subsection{\textsf{SINR} and UAV-BS Association}\label{section:AssociationRule}
The signal-to-interference-plus-noise ratio (\textsf{SINR}) at the typical user when it is associated with a UAV-BS located at $x_o \in \{ \Phi^{\textup {LoS}},\Phi^{\textup {NLoS}}\}$ is given by 
\begin{equation}\label{Eq:SINR}
    \textsf{SINR}=
    \begin{cases}
      \frac{\zeta_N H_{x_o} R_N^{-\alpha_N} }{\sigma^2+I}, & \textup {if} \quad x_o \in \Phi^N \\
      \frac{\zeta_L G_{x_o} R_L^{-\alpha_L} }{\sigma^2+I}, & \textup {if} \quad x_o \in \Phi^L,
    \end{cases}
  \end{equation}
where $R_N$ and $R_L$ are the distances between the serving UAV-BS and the typical user for NLoS and LoS transmissions, respectively, and $\sigma^2$ is the additive white Gaussian noise (AWGN) power. Finally, $I$ is the aggregate interference power defined as
\begin{equation}\label{Eq:Interference}
I=\sum\limits_{x_i\in\phi^N/{x_o}} \zeta_N H_{x_i} D_{N,x_i}^{-\alpha_N}+\sum\limits_{x_i\in\phi^L/{x_o}} \zeta_L G_{x_i} D_{L,x_i}^{-\alpha_L}.
\end{equation}

For the association criteria, we assume that the typical user is associated with the UAV-BS that provides the strongest average \textsf{SINR}. 
The closest UAV-BS does not necessarily provide the strongest \textsf{SINR} due to the differences in path loss parameters between LoS and NLoS transmissions. In particular, a LoS UAV-BS may provide a stronger average \textsf{SINR} than that provided by a closer NLoS UAV-BS due to the fact that $\eta_L>\eta_N$ and $\alpha_L< \alpha_N$. Moreover, an interfering UAV-BS may provide a higher instantaneous \textsf{SINR} for the typical user than that provided by the serving UAV-BS because of a higher small scale fading in comparison to that experienced by the serving UAV-BS.

Based on the strongest average \textsf{SINR} association scheme and the assumption that $\E[H_{x_i}]=\E[G_{x_i}]=1,\hspace{2pt} \forall x_i\in\Phi$, the serving UAV-BS can be written as 
\begin{equation}
x_o=\textup {arg} \max \left\{\eta_N R_N^{-\alpha_N},\eta_L R_L^{-\alpha_L}\right\},
\end{equation}
where $R_N=\min\limits_{\forall x_i\in \Phi^N} D_{N,x_i}$, and $R_L=\min\limits_{\forall x_i\in \Phi^L} D_{L,x_i}$.

\section{Relevant Distance Distributions and Association Probabilities}
In this section, we provide the distribution of the distances between the typical user and the closest UAV-BS for NLoS and LoS transmissions. Furthermore, we characterize the location of the closest interfering NLoS and LoS UAV-BSs given that the typical user is associated with a NLoS or a LoS UAV-BS. Finally, we derive expressions for the association probabilities.
\begin{lma}\label{lma:NearstLoSNLoSDistances}
The probability density function of the distances between the typical user and the closest NLoS and LoS UAV-BSs, denoted by $f_{R_N}(r)$ and $f_{R_L(r)}$, respectively, are given by
\begin{align}
f_{R_N}(r)&=2\pi\lambda r P_N(r)\exp \bigg(-2\pi\lambda\int_{0}^{l(r)}z P_N(z)dz\bigg)\label{Eq:DistrbtionOfR_N}\\
f_{R_L}(r)&=2\pi\lambda r P_L(r)\exp\bigg(-2\pi\lambda\int_{0}^{l(r)}z P_L(z)dz\bigg), \label{Eq:DistrbtionOfR_L}
\end{align}
where $r\geq h$, $l(r)=\sqrt{r^2-h^2}$, $P_N(r)=1-P_L(r)$, and $P_L(r)=P_L(z)|_{z=\sqrt{r^2-h^2}}$.
\end{lma}
\begin{proof}
See Appendix \ref{proof:NearstLoSNLoSDistances}.
\end{proof}

\begin{corl}\label{corl:HorizontalDistance}
Let $Z_N$ and $Z_L$ denote the horizontal distances between the typical user and the projections of the closest NLoS and LoS UAV-BSs on the horizontal plane, respectively. The probability density function of $Z_N$ and $Z_L$, denoted by $f_{Z_N}(z)$ and $f_{Z_L}(z)$, respectively, are given by 
\begin{align}
f_{Z_N}(z)&=2\pi\lambda z P_N(z) \exp\left(-2\pi\lambda\int_{0}^{z}tP_N(t)dt \right)\label{Eq:DistrbtionOfZ_N}\\
f_{Z_L}(z)&=2\pi\lambda z P_L(z)\exp\left(-2\pi\lambda\int_{0}^{z}t P_L(t)dt \right).\label{Eq:DistrbtionOfZ_L}
\end{align}
Proof. \textup {For} $Z_N$, \textup {we have}
\begin{align}\label{Eq:DistrbtonZL}
F_{Z_N}(z)&=\Prb(Z_N\leq z)\overset{(a)}{=}F_{R_N}(\sqrt{z^2+h^2}) \nonumber\\
&\overset{(b)}{=}1-\exp\left(-2\pi\lambda\int_{0}^{z}t P_N(t) dt \right ),
\end{align}
\end{corl}
where (a) is due to $Z_N=\sqrt{R_N^2-h^2}$, and (b) follows from (\ref{Eq:proofNearestNLoSDis}). Finally, we complete the proof by taking the derivative of $F_{Z_N}(z)$ with respect to $z$. Following the same steps, we arrive at the final result for $f_{Z_L}(z)$.

The following remarks give clear insight on the range over which the interfering UAV-BSs are located which will be useful when we present the main results of this paper.
\begin{remarkl}\label{rmrk:ClostNLoSIntrfrer}
Given that the typical user is associated with a NLoS UAV-BS located at a distance $r$ from the typical user, the closest interfering LoS UAV-BS is at least at a distance
\begin{equation}
d_L=\left(\frac{\eta_L}{\eta_N}\right)^{\frac{1}{\alpha_L}}r^{\frac{\alpha_N}{\alpha_L}}. 
\end{equation}
\end{remarkl}
\begin{remarkl}\label{rmrk:ClostLoSIntrfrer}
Given that the typical user is associated with a LoS UAV-BS located at a distance $r$ from the typical user, the closest interfering NLoS UAV-BS is at least at a distance 
\begin{equation}
d_N=
\begin{cases}
h, \hspace{65pt} \textup {if} \quad h \leq r \leq \left(\frac{\eta_L}{\eta_N}\right)^{\frac{1}{\alpha_L}}h^{\frac{\alpha_N}{\alpha_L}}\\
\left(\frac{\eta_N}{\eta_L}\right)^{\frac{1}{\alpha_N}}r^{\frac{\alpha_L}{\alpha_N}},\hspace{11pt} \textup {if} \quad r> \left(\frac{\eta_L}{\eta_N}\right)^{\frac{1}{\alpha_L}}h^{\frac{\alpha_N}{\alpha_L}}. 
\end{cases}
\end{equation}
\end{remarkl}
As per the association rule in section \ref{section:AssociationRule}, the typical user is associated with a single UAV-BS which could be a LoS or a NLoS UAV-BS. The following lemma gives the probabilities that the typical user is either associated with a LoS UAV-BS or a NLoS UAV-BS.
\begin{lma}\label{lma:LoSandNLoSAssociation}
The probability that the typical user is associated with a LoS UAV-BS is given by
\begin{align}
A_L&=1-2\pi\lambda \int_{0}^{\infty}z P_N(z)\exp\left(-2\pi\lambda\int_{0}^{\sqrt{U(z)}}tP_L(t)dt\right)\nonumber\\
&\hspace{0pt}\times \exp\left( -2\pi\lambda \int_{0}^{z}t P_N(t)dt\right)dz,
\end{align}
where $U(z)=\left(\frac{\eta_L}{\eta_N}\right)^{\frac{2}{\alpha_L}}\left(z^2+h^2\right)^{\frac{\alpha_N}{\alpha_L}}-h^2$. The probability that the typical user is associated with a NLoS UAV-BS is ${A_N=1-A_L}$. 
\end{lma}

\begin{proof}
See Appendix \ref{proof:LoSandNLoSAssociation}.
\end{proof}

\section{Coverage probability}
The coverage probability is generally defined as the probability that the \textsf{SINR} is greater than a designated threshold $T$:
\begin{equation}
P_C=\Prb(\textsf{SINR}> T).
\end{equation}

We begin this section by deriving the Laplace transform of the interference, which is given in Lemma \ref{lma:LaplaceTransform}.
\begin{lma}\label{lma:LaplaceTransform}
The Laplace transform of the aggregated interference power conditioned on the serving UAV-BS being at a distance $r$ from the typical user is given by
\begin{align}\label{Eq:LaplaceTransform}
&\mathcal{L}_I(s|r)=\exp\Bigg( -2\pi\lambda \int_{v_1(r)}^{\infty} \left[ 1-\frac{1}{1+s \zeta_N (t^2+h^2)^{\frac{-\alpha_N}{2}}}\right]\nonumber\\
&\times t P_N(t)dt -2\pi\lambda \int_{v_2(r)}^{\infty}\left[ 1-\left(\frac{m}{m+s \zeta_L (t^2+h^2)^{\frac{-\alpha_L}{2}}}\right)^m\right]\nonumber\\
&\times t P_L(t)dt \Bigg),
\end{align}
with
\begin{align}
v_1(r)&=\sqrt{r^2-h^2}, \hspace{5pt}v_2(r)= \sqrt{d_L^2-h^2} \hspace{6pt} \textup {if} \quad x_o \in \Phi^N\nonumber \\
v_1(r)&=\sqrt{d_N^2-h^2},\hspace{2pt} v_2(r)= \sqrt{r^2-h^2}  \hspace{7pt} \textup {if} \quad x_o \in \Phi^L\nonumber.
\end{align}
\end{lma}
\begin{proof}
See Appendix \ref{Proof:LaplaceTransform}.
\end{proof}

Now that we have developed expressions for association probabilities and the Laplace transform of the interference, we present the main theorem on the coverage probability.
\begin{thm}\label{thm:coverage}
The probability of coverage $P_C$ is given by
\begin{equation}
P_C=P_{C,L}A_L+P_{C,N}A_N,
\end{equation}
where $P_{C,L}$ and $P_{C,N}$ are the conditional coverage probabilities given that the typical user is associated with a LoS UAV-BS or a NLoS UAV-BS, respectively, and are given by 

{\small\begin{align}\label{Eq:CondCovgeLoS}
&P_{C,L}=\sum\limits_{k=1}^{m} \binom{m}{k}(-1)^{k+1} \int_{h}^{\infty} \exp\Bigg(-k\mu_L\sigma^2 r^{\alpha_L}-2\pi\lambda\nonumber\\
&\int_{l(d_N)}^{\infty}\left[1-\frac{1}{1+k\mu_L r^{\alpha_L}\zeta_N (t^2+h^2)^{\frac{-\alpha_N}{2}}}\right]
t P_N(t)dt\\
&-2\pi\lambda\int_{l(r)}^{\infty}\left[1-\left(\frac{m}{m+k\mu_L r^{\alpha_L} \zeta_L (t^2+h^2)^{\frac{-\alpha_L}{2}}}\right)^m\right]\nonumber\\
&t P_L(t)dt  \Bigg) f_{R_L}(r)dr,\nonumber
\end{align}
}%
and
{\small\begin{align}\label{Eq:CondCovgeNLoS}
&P_{C,N}=\int_{h}^{\infty}\exp\Bigg(-\sigma^2 T \zeta_N^{-1} r^{\alpha_N}-2\pi\lambda\int_{l(r)}^{\infty} \Bigg[1-\nonumber\\
&\frac{1}{1+ T r^{\alpha_N}(t^2+h^2)^{\frac{-\alpha_N}{2}}}\Bigg]tP_N(t)dt-2\pi\lambda\int_{l(d_L)}^{\infty}\Bigg[1-\\
&\Big(\frac{m}{m+\eta_N^{-1}\eta_L T r^{\alpha_N}(t^2+h^2)^{\frac{-\alpha_L}{2}}}\Big)^m\Bigg]t P_L(t)dt\Bigg) f_{R_N}(r)dr,\nonumber
\end{align}
}%
where $\mu_L=\alpha m T \zeta_L^{-1}$, $l(d_N)=\sqrt{d_N^2-h^2}$, $l(r)=\sqrt{r^2-h^2}$, and $l(d_L)=\sqrt{d_L^2-h^2}$.
\end{thm}
\begin{proof}
See Appendix \ref{prf:CoverageProbability}. 
\end{proof}

\section{average achievable rate}
The average achievable rate of the typical user is given by ${\tau=\E [\ln(1+\textsf{SINR})]}$ (nats/Hz), where 1 bit = ln(2) = 0.693 nats \cite{andrews2011tractable}. The following theorem presents the main rate theorem.
\begin{thm}\label{thm:Rate}
The average downlink rate of a typical user is given by
\begin{equation}
\tau=\tau_L A_L+\tau_N A_N,
\end{equation}
where $\tau_L$ and $\tau_N$ are the average achievable rates given that the typical user is associated with a LoS or a NLoS UAV-BS, respectively, and are given by 
{\small\begin{equation}\label{Eq:CondRtaeLoS}
\begin{aligned}
&\tau_L=\sum\limits_{k=1}^{m} \binom{m}{k}(-1)^{k+1} \int_{r\geq h} \exp\left(k \rho_L \sigma^2 r^{\alpha_L}\right) \\
&\int_{y >0}  \exp\Bigg(-k \rho_L \sigma^2 r^{\alpha_L} e^y -2\pi\lambda \int_{l(d_N)}^{\infty}\\
&\left[ 1-\frac{1}{1+k \rho_L r^{\alpha_L}(e^y-1) \zeta_N (t^2+h^2)^{\frac{-\alpha_N}{2}}}\right] t P_N(t)dt \\
 &-2\pi\lambda \int_{l(r)}^{\infty}\left[ 1-\left(\frac{m}{m+k \alpha m r^{\alpha_L}(e^y-1) (t^2+h^2)^{\frac{-\alpha_L}{2}}}\right)^m\right]\\
 & t P_L(t)dt \Bigg)  f_{R_L}(r) dr,
 \end{aligned}
\end{equation}}%
and
{\small\begin{equation}\label{Eq:CondRtaeNLoS}
\begin{aligned}
&\tau_N=\int\limits_{r\geq h} \exp\left(\sigma^2 \zeta_N^{-1} r^{\alpha_N} \right)\ \int\limits_{y > 0}\exp\Bigg(-\sigma^2 \zeta_N^{-1} r^{\alpha_N} e^y\\
& -2\pi\lambda \int\limits_{l(r)}^{\infty} \left[ 1-\frac{1}{1+r^{\alpha_N} (e^y-1)(t^2+h^2)^{\frac{-\alpha_N}{2}}}\right] t P_N(t)dt\\
 &-2\pi\lambda \int\limits_{l(d_L)}^{\infty}\left[ 1-\left(\frac{m}{m+\zeta_N^{-1}r^{\alpha_N} (e^y-1) \zeta_L (t^2+h^2)^{\frac{-\alpha_L}{2}}}\right)^m\right]\\
&\times t P_L(t)dt \Bigg),
\end{aligned}
\end{equation}}%
where $\rho_L=\alpha m \zeta_L^{-1}$, $l(d_N)=\sqrt{d_N^2-h^2}$, $l(r)=\sqrt{r^2-h^2}$, and $l(d_L)=\sqrt{d_L^2-h^2}$.
\end{thm}
\begin{proof}
See Appendix \ref{proof:AverageRate}.
\end{proof}

\section{numerical results}
In this section, we provide simulations to evaluate our main analytical results. In particular, we use MATLAB to simulate Theorem \ref{thm:coverage} and Theorem \ref{thm:Rate}. We consider
 dense urban area with parameters $ a= 12.08, b = 0.11, \eta_L= 0.69$ and $\eta_N = 0.005$ \cite{alzenad20173,Hourani}. We also consider UAV-BSs that transmit their signals at $f_c=2$ GHz and $P_t=30$ dBm while the noise power is assumed $-174$ dBm/Hz. The NLoS and LoS path loss exponents are $\alpha_N=3.5$ and $\alpha_L=2$. The shape parameter of the Nakagami-m fading is $m=3$ and the system bandwidth is 10 MHz.

The impact of the UAV-BSs altitude on the coverage probability of a typical user is studied in Fig. \ref{fig:CovgeSINR}. It can be seen from Fig. \ref{fig:CovgeSINR} that as the UAV-BSs altitude increases, the coverage probability decreases due to the increase in path loss. It can also be observed from Fig. \ref{fig:CovgeSINR} that the analytical results in Theorem \ref{thm:coverage} match the simulations with negligible errors. 
\begin{figure}
\begin{center}
\includegraphics[ height=6.5cm, width=9cm]{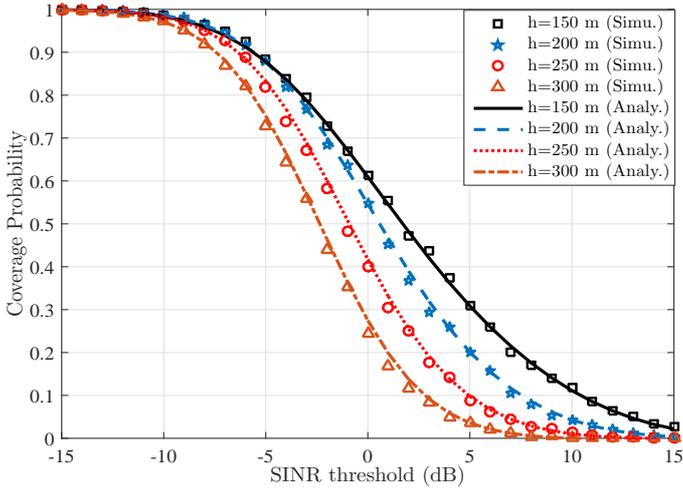}
\caption{\footnotesize Coverage probability versus \textsf{SINR} threshold for the typical user for different altitudes.}
\label{fig:CovgeSINR}
\end{center}
\end{figure}

The impact of the UAV-BSs altitude and their densities on the data rate achieved by a typical user are studied in Fig.~\ref{fig:rate}, where we plot the average rate versus UAV-BSs altitude for the UAV-BSs densities $\lambda=3,5,7$ and $9$ BSs/$\textup {km}^2$. The results in Fig. \ref{fig:rate} show that for a given density, the average achievable rate degrades as the UAV-BSs altitude increases due to the increase in the path loss. Furthermore, for a given UAV-BS altitude, the average achievable rate decreases as the UAV-BSs density increases. This is because the interfering UAV-BSs become closer to the typical user as the density increases which degrades the \textsf{SINR} at the typical user.

\section{conclusion}
 In this paper, we proposed a stochastic geometry framework to analyze coverage and rate in a network of UAV-BSs deployed at a particular height. The framework accommodates both LoS and NLoS transmissions, and considers Rayleigh fading and Nakagami-m fading for NLoS and LoS links, respectively. 
 We derived analytical expressions for the conditional Laplace transform of the interference power, the association probabilities, and the distribution of the distances between the typical user and the closest NLoS and LoS UAV-BSs. Approximate expressions for the coverage probability and average achievable rate were also derived. Interestingly, we showed that these approximations match the simulations with negligible errors.  

\begin{appendices}

\section{proof of lemma \ref{lma:NearstLoSNLoSDistances}}\label{proof:NearstLoSNLoSDistances}
Given that $R_N$ is a random variable, 
the corresponding horizontal Euclidean distance
, denoted by $Z_N$, is also a random variable given by $Z_N=\sqrt{R_N^2-h^2}$. The cumulative distribution function (CDF) of $R_N$ is given by 
\begin{align}\label{Eq:proofNearestNLoSDis}
F_{R_N}(r)&=1-\Prb(R_N> r)=1-\Prb\left(Z_N > \sqrt{r^2-h^2}\right)\nonumber\\
&\overset{a}{=}1-\exp\left(-2\pi\lambda\int_{0}^{\sqrt{r^2-h^2}}z P_N(z)dz\right),
\end{align}
where (a) follows from the null probability of the PPP \cite{andrews2011tractable}. Finally, $f_{R_N}(r)=\dfrac{d}{dr}F_{R_N}(r)$ which completes the proof of $f_{R_N}(r)$. By following the same steps as for $f_{R_N}(r)$, we can complete the proof of $f_{R_L}(r)$.
\begin{figure}
\begin{center}
\includegraphics[ height=6.5cm, width=9cm]{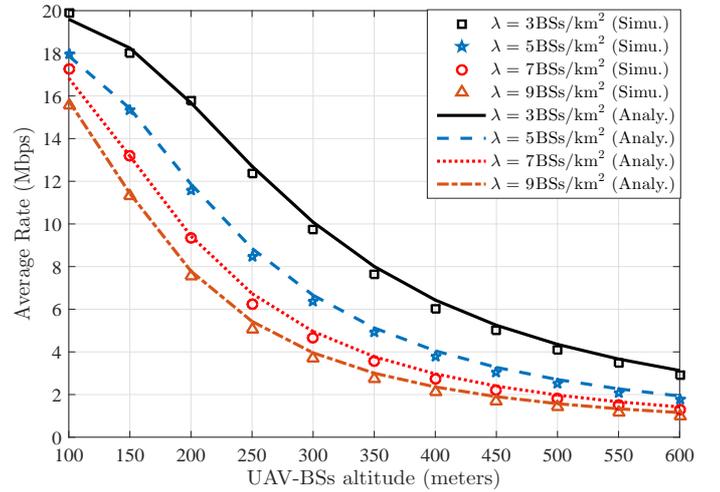}
\caption{\footnotesize Average rate versus UAV-BS altitude for the typical user for different UAV-BSs densities (BW=10 MHz).}
\label{fig:rate}
\end{center}
\end{figure}
\section{proof of Lemma  \ref{lma:LoSandNLoSAssociation}}\label{proof:LoSandNLoSAssociation}
Since the UAV-BS that provides the strongest average \textsf{SINR} also provides the strongest average received power \cite{yang2017density}, the probability that the typical user is associated with a LoS UAV-BS is then given by
\begin{align}
A_L&=\Prb \left(\zeta_L R_L^{-\alpha_L}>\zeta_N  R_N^{-\alpha_N}\right)\nonumber\\
&\overset{(a)}{=}\Prb\left(Z_L^2<(\frac{\eta_L}{\eta_N})^{\frac{2}{\alpha_L}}(Z_N^2+h^2)^{\frac{\alpha_N}{\alpha_L}}-h^2\right)\nonumber\\
&\overset{(b)}{=}\int_{0}^{\infty}\left(1-\Prb \left(Z_L^2>U(z)\right)\right)f_{Z_N}(z)dz\nonumber\\
&\overset{(c)}{=}1-\int_{0}^{\infty}\Prb \left(Z_L>\sqrt{U(z)}\right)f_{Z_N}(z)dz\nonumber\\
&\overset{(d)}{=}1-\int_{0}^{\infty} \exp \left(-2\pi\lambda\int_{0}^{\sqrt{U(z)}}tP_L(t)dt\right) f_{Z_N}(z)dz, 
\end{align}
where (a) is due to $R_L=\sqrt{Z_L^2+h^2}$ and ${R_N=\sqrt{Z_N^2+h^2}}$, (b) follows from conditioning on $Z_N=z$, and ${U(z)=\left(\frac{\eta_L}{\eta_N}\right)^{\frac{2}{\alpha_L}}\left(z^2+h^2\right)^{\frac{\alpha_N}{\alpha_L}}-h^2}$, (c) follows from the fact that $Z_N$ is a positive random variable, and (d) follows from the null probability of the PPP. Finally, by substituting  (\ref{Eq:DistrbtionOfZ_N}) into (d), we complete the proof.
\section{proof of Lemma \ref{lma:LaplaceTransform}}\label{Proof:LaplaceTransform}
The Laplace transform $\mathcal{L}_I(s)$ can be written as follows:
\begin{align}
&\mathcal{L}_I(s)=\E_I\left[\exp\left(-sI\right)\right]\nonumber\\
& \hspace{22pt}\overset{(a)}{=}\E_{\phi^N}\bigg[\prod\limits_{x_i\in\phi^N\backslash x_o} \E_H \exp\Big(-s \zeta_N H_{x_i} D_{N,x_i}^{-\alpha_N}\Big)\bigg]  \nonumber\\
 & \hspace{30pt}\times \E_{\phi^L}\bigg[\prod\limits_{x_i\in\phi^L\backslash x_o} \E_G \exp\Big(-s \zeta_L G_{x_i} D_{L,x_i}^{-\alpha_L}\Big)\bigg]\nonumber\\
 & \hspace{22pt}\overset{}{=}\E_{\phi^N}\bigg[\prod\limits_{x_i\in\phi^N\backslash x_o}\frac{1}{1+s \zeta_N D_{N,x_i}^{-\alpha_N}}\bigg] \nonumber\\
 & \hspace{30pt}\times \E_{\phi^L}\bigg[\prod\limits_{x_i\in\phi^L\backslash x_o} \bigg(\frac{m}{m+s \zeta_L D_{L,x_i}^{-\alpha_L}}\bigg)^m\bigg],
 \end{align}
where (a) follows from (\ref{Eq:Interference}), the i.i.d distribution and the independence of the spatial point process and small scale fading.

Now given that the typical user is associated with a NLoS UAV-BS (i.e., $x_o\in \Phi^N$) located at a distance $R_N=r$, and from Remark (\ref{rmrk:ClostNLoSIntrfrer}) and the probability generating functional (PGFL) of the PPP, we obtain the final result in (\ref{Eq:LaplaceTransform}) for the case $x_o\in \Phi^N$. Similarly, the conditional Laplace transform of the aggregated interference power when the typical user is associated with a LoS UAV-BS can be obtained from remark (\ref{rmrk:ClostLoSIntrfrer}) and the PGFL of the PPP.

\section{proof of theorem \ref{thm:coverage}}\label{prf:CoverageProbability}
Given that the typical user is associated with a LoS UAV-BS, the conditional coverage probability  $P_{C,L}$ is given by
{\small\begin{align}
P_{C,L}&=\Prb\left(\frac{\zeta_L G_{x_o} R_L^{-\alpha_L} }{\sigma^2+I}>T\right)\nonumber\\
&\overset{(a)}{=}\int_{h}^{\infty}\Prb \left( G_{x_o}>T \zeta_L^{-1}r^{\alpha_L}\left(\sigma^2+I\right) \right) f_{R_L}(r) dr\nonumber\\
&\overset{(b)}{=} 1-\int_{h}^{\infty} \E_I\left[ F_G\left(T \zeta_L^{-1}r^{\alpha_L}\left(\sigma^2+I\right) \right)\right] f_{R_L}(r) dr\nonumber\\
&\overset{(c)}{=} 1-\int_{h}^{\infty} \E_I\left[ \frac{\Gamma_l\left(m,mT \zeta_L^{-1}r^{\alpha_L}\left(\sigma^2+I\right) \right)}{\Gamma(m)}\right] f_{R_L}(r) dr,
\end{align}}%
where (a) follows from conditioning on the serving LoS UAV-BS being at a distance $R_L=r$ from the typical user, (b) follows from the definition $F_G(g) = \Prb(G\leq g)$ and taking the conditional expectation with respect to interference, and (c) follows from the definition of the CDF of Gamma distribution $F_G(g)=\frac{\Gamma_l (m,mg)}{\Gamma (m)}$ where $\Gamma_l (m,mg)=\int_{0}^{mg}t^{m-1} e^{-t}dt$ is the lower incomplete gamma function.

The evaluation of the CDF of Gamma distribution requires evaluating higher order derivatives of the Laplace transform. The larger the shape parameter $m$ is, the higher the evaluation complexity is. Therefore, in the following, we provide an approximate evaluation of the coverage probability. In particular, we use a tight bound for the CDF of the Gamma distribution rather than using the exact evaluation. The CDF of Gamma distribution can be bounded as \cite{alzer1997some}
\begin{equation}\label{Eq:GammaBound}
\big(1-e^{-\beta mg}\big)^m<\frac{\Gamma_l\left(m,mg\right)}{\Gamma(m)}<\left(1-e^{-\alpha mg}\right)^m,
\end{equation}
where $m\neq 1$, and
\begin{equation}
\beta=
\begin{cases}
1, \hspace{27pt} \textup{if}\quad m>1\\
(m!)^{\frac{-1}{m}}, \hspace{2pt} \textup{if}\quad m<1\\
\end{cases}
\alpha=
\begin{cases}
(m!)^{\frac{-1}{m}}, \hspace{2pt} \textup{if}\quad m>1\\
1, \hspace{27pt} \textup{if}\quad m<1.
\end{cases}
\end{equation}

It has been shown in \cite{bai2015coverage} that the upper bound in (\ref{Eq:GammaBound}) provides a good approximation to the CDF of Gamma distribution. Therefore, we use the tighter upper bound. 
The conditional coverage probability can then be written as
{\small\begin{align}
P_{C,L}\approx &1-\int_{h}^{\infty} \E_I\left[ \left(1-\exp\left(-\mu_L r^{\alpha_L}\left(\sigma^2+I\right) \right)\right)^m\right]\nonumber\\
&\hspace{0pt} \times f_{R_L}(r) dr\nonumber\\
\overset{(a)}{=}& \int_{h}^{\infty} \E_I\left[ \sum\limits_{k=1}^{m} \binom{m}{k}(-1)^{k+1} \exp\left(-k \mu_L r^{\alpha_L}\left(\sigma^2+I\right) \right)\right]\nonumber\\
&\hspace{0pt} \times f_{R_L}(r) dr\nonumber\\
\overset{(b)}{=}& \int_{h}^{\infty}  \sum_{k=1}^{m} \binom{m}{k}(-1)^{k+1} \exp\left(-k\mu_L\sigma^2 r^{\alpha_L}\right) \nonumber\\
&\hspace{0pt} \times \E_I \left[ \exp\left(-k\mu_L  r^{\alpha_L} I\right) \right] f_{R_L}(r) dr\nonumber\\
=& \sum\limits_{k=1}^{m} \binom{m}{k}(-1)^{k+1} \int_{h}^{\infty} \exp\left(-k\mu_L\sigma^2 r^{\alpha_L}\right) \nonumber\\
&\hspace{0pt} \times \mathcal{L}_I\left(k\mu_L r^{\alpha_L}\right) f_{R_L}(r) dr,
\end{align}}%
where $\mu_L=\alpha m T \zeta_L^{-1}$, (a) follows from the binomial theorem and the assumption that $m$ is an integer, and (b) results from the linearity of the expectation. Finally, $\mathcal{L}_I\left(k\mu_L\sigma^2 r^{\alpha_L}\right)$ is obtained from (\ref{Eq:LaplaceTransform}) where $v_1(r)=\sqrt{d_N^2-h^2}$ and ${v_2(r)=\sqrt{r^2-h^2}}$ which completes the proof of (\ref{Eq:CondCovgeLoS}). Similarly, 
$P_{C,N}$ can be derived by following the same approach as that of $P_{C,L}$ and by setting $m=1$. Therefore, we omit the detailed proof of (\ref{Eq:CondCovgeNLoS}). 

\section{proof of theorem \ref{thm:Rate}}\label{proof:AverageRate}
The average achievable rate is given by 
\begin{align}
\tau=&\E \left[\ln \left(1+\textsf{SINR}\right)\right]\overset{(a)}{=}\int_{y>0}\Prb\left(\ln \left(1+ \textsf{SINR}\right)>y\right)dy\nonumber\\
\overset{(b)}{=}&\int_{y>0}\Prb \left(\ln \left(1+ \textsf{SINR}\right)>y|x_o\in\Phi^L\right)dy A_L+\nonumber\\
&\int_{y>0}\Prb\left(\ln \left(1+ \textsf{SINR}\right)>y|x_o\in\Phi^N\right)dy A_N\nonumber\\
=&\tau_L A_L + \tau_N A_N,
\end{align}
where (a) follows from the fact that for a positive random variable $X$, we have $\E[X]=\int_{y>0}\Prb(X>y)dy$ \cite{andrews2011tractable}, and (b) follows from the law of total probability and linearity of integrals. 
Now, given that the typical user is associated with a LoS UAV-BS
, the conditional average rate is given by
{\small\begin{align}
\tau_L\overset{(a)}{=}&\int_{y>0}\Prb\left(\ln \left(1+ \frac{\zeta_L G_{x_o} R_L^{-\alpha_L} }{\sigma^2+I}\right)>y\right)dy\nonumber\\
\overset{(b)}{=}&\int\limits_{y>0}\int\limits_{r \geq h} \E_I \left[\Prb\left(G_{x_o}> \zeta_L^{-1} r^{\alpha_L}(e^y-1) (\sigma^2+I)\right)\right]\nonumber\\
& \times f_{R_L}(r)dr dy\nonumber\\
\overset{(c)}{=}&\int\limits_{y >0} \int\limits_{r\geq h} \E_I\left[ 1-\left(1-\exp\left(-\rho_L r^{\alpha_L}(e^y-1) (\sigma^2+I) \right)\right)^m\right]\nonumber\\ 
&\times f_{R_L}(r) dr dy\nonumber\\
\overset{(d)}{=}&\int\limits_{y >0} \int\limits_{r\geq h} \E_I \Bigg[ \sum\limits_{k=1}^{m} \binom{m}{k}(-1)^{k+1} \nonumber\\
&\exp\left(-k \rho_L r^{\alpha_L}(e^y-1) (\sigma^2+I)\right) \Bigg] f_{R_L}(r) dr dy\nonumber\\
\overset{(e)}{=}&\sum\limits_{k=1}^{m} \binom{m}{k}(-1)^{k+1} \int\limits_{y >0} \int\limits_{r\geq h} \exp\left(-k \rho_L \sigma^2 r^{\alpha_L}(e^y-1)\right)\nonumber\\
 &\E_I\left[\exp\left(-k \rho_L r^{\alpha_L}(e^y-1)I\right) \right] f_{R_L}(r) dr dy\nonumber\\
 \overset{(f)}{=}&\sum\limits_{k=1}^{m} \binom{m}{k}(-1)^{k+1} \int\limits_{r\geq h} \exp\left(k \rho_L \sigma^2 r^{\alpha_L}\right) \nonumber\\
 \times\int\limits_{y >0}  &\exp\left(-k \rho_L \sigma^2 r^{\alpha_L}e^y \right)\mathcal{L}_I \left(k \rho_L r^{\alpha_L}(e^y-1)\right) dy f_{R_L}(r) dr,\nonumber
\end{align}}%
where $\rho_L=\alpha m \zeta_L^{-1}$, (a) follows from (\ref{Eq:SINR}), (b) follows from conditioning on $R_L=r$ and taking the conditional expectation with respect to interference, (c) is from the upper bound of Gamma distribution given in (\ref{Eq:GammaBound}), (d) results from the binomial theorem and the assumption that $m$ is an integer, (e) follows from the linearity of integrals and expectation, and the last step results from swapping the integration orders. Finally, plugging (\ref{Eq:LaplaceTransform}) into (f) when $s=k \rho_L r^{\alpha_L}(e^y-1)$ completes the proof of (\ref{Eq:CondRtaeLoS}). The average achievable rate given that the typical user is associated with a NLoS UAV-BS ($\tau_N$) can be derived by following the same approach as that of $\tau_L$ and by setting $m=1$. Therefore, we omit the detailed proof of (\ref{Eq:CondRtaeNLoS}).

\end{appendices}

\bibliographystyle{IEEEtran}
\bibliography{IEEEfull,RefList}

\begin{thebibliography}{10}
\providecommand{\url}[1]{#1}
\csname url@samestyle\endcsname
\providecommand{\newblock}{\relax}
\providecommand{\bibinfo}[2]{#2}
\providecommand{\BIBentrySTDinterwordspacing}{\spaceskip=0pt\relax}
\providecommand{\BIBentryALTinterwordstretchfactor}{4}
\providecommand{\BIBentryALTinterwordspacing}{\spaceskip=\fontdimen2\font plus
\BIBentryALTinterwordstretchfactor\fontdimen3\font minus
  \fontdimen4\font\relax}
\providecommand{\BIBforeignlanguage}[2]{{%
\expandafter\ifx\csname l@#1\endcsname\relax
\typeout{** WARNING: IEEEtran.bst: No hyphenation pattern has been}%
\typeout{** loaded for the language `#1'. Using the pattern for}%
\typeout{** the default language instead.}%
\else
\language=\csname l@#1\endcsname
\fi
#2}}
\providecommand{\BIBdecl}{\relax}
\BIBdecl

\bibitem{cao2018airborne}
X.~Cao, P.~Yang, M.~Alzenad, X.~Xi, D.~Wu, and H.~Yanikomeroglu, ``Airborne
  communication networks: A survey,'' \emph{{IEEE J.\ Select.\ Areas Commun.}},
  vol.~PP, no.~99, pp. 1--1, 2018, {DOI}:10.1109/JSAC.2018.2864423.

\bibitem{irem2018spatial}
I.~Bor-Yaliniz, A.~El-Keyi, and H.~Yanikomeroglu, ``Spatial configuration of
  agile wireless networks with drone-{BS}s and user-in-the-loop,'' \emph{to
  appear in IEEE Trans. Wireless Commun.}, vol.~PP, no.~99, pp. 1--1.

\bibitem{alzenad2018fso}
M.~Alzenad, M.~Z. Shakir, H.~Yanikomeroglu, and M.-S. Alouini, ``{FSO}-based
  vertical backhaul/fronthaul framework for {5G+} wireless networks,''
  \emph{{IEEE Commun. Mag.}}, vol.~56, no.~1, pp. 218--224, 2018.

\bibitem{zeng2016wireless}
Y.~Zeng, R.~Zhang, and T.~J. Lim, ``Wireless communications with unmanned
  aerial vehicles: Opportunities and challenges,'' \emph{{IEEE Commun. Mag.}},
  vol.~54, no.~5, pp. 36--42, May 2016.

\bibitem{IremMagazine}
I.~Bor-Yaliniz and H.~Yanikomeroglu, ``The new frontier in {RAN} heterogeneity:
  Multi-tier drone-cells,'' \emph{{IEEE Commun. Mag.}}, vol.~54, no.~11, pp.
  48--55, Nov. 2016.

\bibitem{alzenad20173}
M.~Alzenad, A.~El-Keyi, F.~Lagum, and H.~Yanikomeroglu, ``3{D} placement of an
  unmanned aerial vehicle base station ({UAV-BS}) for energy-efficient maximal
  coverage,'' \emph{{IEEE} Wireless Commun. Lett.}, vol.~6, no.~4, pp.
  434--437, Aug. 2017.

\bibitem{alzenad20173d}
M.~Alzenad, A.~El-Keyi, and H.~Yanikomeroglu, ``3{D} placement of an unmanned
  aerial vehicle base station for maximum coverage of users with different
  {QoS} requirements,'' \emph{{IEEE} Wireless Commun. Lett.}, vol.~7, no.~1,
  pp. 38--41, Feb. 2018.

\bibitem{Elham2017Backhaul}
E.~Kalantari, M.~Z. Shakir, H.~Yanikomeroglu, and A.~Yongacoglu,
  ``Backhaul-aware robust 3{D} drone placement in 5{G}+ wireless networks,''
  \emph{\textup{in} Proc. IEEE Int. Conf. Commun. Workshop (ICCW)}, Paris,
  France, May 2017.

\bibitem{MozaffariEfficient}
\vspace{0mm}M. Mozaffari, W.~Saad, M.~Bennis, and M.~Debbah, ``Efficient
  deployment of multiple unmanned aerial vehicles for optimal wireless
  coverage,'' \emph{{IEEE Commun.\ Lett.}}, vol.~20, no.~8, pp. 1647--1650,
  Aug. 2016.

\bibitem{galkin2017coverage}
B.~Galkin, J.~Kibi{\l}da, and L.~A. DaSilva, ``Coverage analysis for
  low-altitude {UAV} networks in urban environments,'' in \emph{{Proc. IEEE
  Glob. Commun. Conf.~(Globecom)}}, Singapore, Dec. 2017.

\bibitem{chetlur2017downlink}
V.~V. Chetlur and H.~S. Dhillon, ``Downlink coverage analysis for a finite
  {3-D} wireless network of unmanned aerial vehicles,'' \emph{{IEEE Trans.\
  Commun.}}, vol.~65, no.~10, pp. 4543--4558, Oct. 2017.

\bibitem{zhang2017spectrum}
C.~Zhang and W.~Zhang, ``Spectrum sharing for drone networks,'' \emph{{IEEE J.\
  Select.\ Areas Commun.}}, vol.~35, no.~1, pp. 136--144, Jan. 2017.

\bibitem{zhou2018uplink}
X.~Zhou, J.~Guo, S.~Durrani, and H.~Yanikomeroglu, ``Uplink coverage
  performance of an underlay drone cell for temporary events,'' in \emph{{Proc.
  IEEE Int. Conf. Commun. (ICC)}}, Kansas City, USA, May 2018.

\bibitem{Hourani}
A.~Al-Hourani, S.~Kandeepan, and S.~Lardner, ``Optimal {LAP} altitude for
  maximum coverage,'' \emph{{IEEE} Wireless Commun. Lett.}, vol.~3, no.~6, pp.
  569--572, Dec. 2014.

\bibitem{wackerly2007mathematical}
D.~Wackerly, W.~Mendenhall, and R.~Scheaffer, \emph{Mathematical Statistics
  with Applications}, 2007.

\bibitem{andrews2011tractable}
J.~G. Andrews, F.~Baccelli, and R.~K. Ganti, ``A tractable approach to coverage
  and rate in cellular networks,'' \emph{{IEEE Trans.\ Commun.}}, vol.~59,
  no.~11, pp. 3122--3134, Nov. 2011.

\bibitem{yang2017density}
B.~Yang, G.~Mao, M.~Ding, X.~Ge, and X.~Tao, ``Dense small cell networks: From
  noise-limited to dense interference-limited,'' \emph{{IEEE Trans. Veh.
  Technol.}}, vol.~67, no.~5, pp. 4262--4277, May 2018.

\bibitem{alzer1997some}
H.~Alzer, ``On some inequalities for the incomplete {G}amma function,''
  \emph{Math. Comput.}, vol.~66, no. 218, pp. 771--778, 1997.

\bibitem{bai2015coverage}
T.~Bai and R.~W. Heath, ``Coverage and rate analysis for millimeter-wave
  cellular networks,'' \emph{{IEEE Trans.\ Wireless Commun.}}, vol.~14, no.~2,
  pp. 1100--1114, Feb. 2015.

\end{thebibliography}
\end{document}